\documentclass{dac}
\usepackage{graphicx} 
\usepackage{hyperref}
\usepackage{amssymb}
\usepackage{mathtools}
\usepackage{algorithm}
\usepackage{algpseudocode}
\usepackage{braket}
\usepackage{graphicx}
\usepackage{xcolor}
\usepackage{xspace}
\usepackage{hypbmsec}
\usepackage{hyperref}

\algtext*{EndWhile}
\algtext*{EndIf}

\DeclareMathOperator{\CountSuffixes}{CS}

\DeclareMathOperator{\emptyword}{\varepsilon}

\title{Ranking and Unranking k-subsequence universal words}

\author{Duncan Adamson\\
  \emph{Leverhulme Research Centre for Functional Materials Design, The University of Liverpool, UK\footnote{This work was completed at, and partially funded by the University of G\"ottingen. 4}}\\
  \texttt{d.a.adamson@liverpool.ac.uk}} 

\date{}
\begin{document}
\maketitle

\begin{abstract}
A subsequence of a word $w$ is a word $u$ such that $u = w[i_1] w[i_2] , \dots w[i_{|u|}]$, for some set of indices $1 \leq i_1 < i_2 < \dots < i_k \leq |w|$.
A word $w$ is $k$-subsequence universal over an alphabet $\Sigma$ if every word in $\Sigma^k$ appears in $w$ as a subsequence.
In this paper, we provide new algorithms for $k$-subsequence universal words of fixed length $n$ over the alphabet $\Sigma = \{1,2,\dots, \sigma\}$.
Letting $\mathcal{U}(n,k,\sigma)$ denote the set of $n$-length $k$-subsequence universal words over $\Sigma$, we provide:
\begin{itemize}
	\item an $O(n k \sigma)$ time algorithm for counting the size of $\mathcal{U}(n,k,\sigma)$;
	\item an $O(n k \sigma)$ time algorithm for ranking words in the set $\mathcal{U}(n,k,\sigma)$;
	\item an $O(n k \sigma)$ time algorithm for unranking words from the set $\mathcal{U}(n,k,\sigma)$;
	\item an algorithm for enumerating the set $\mathcal{U}(n,k,\sigma)$ with $O(n \sigma)$ delay after $O(n k \sigma)$ preprocessing.
\end{itemize}
\end{abstract}


\section{Introduction}
\label{sec:intro}

Words and subsequences are two fundamental combinatorial objects.
Informally, a subsequence of a word $w$ is a word $u$ that can be found by deleting some subset of the symbols $w$.
Subsequences are a heavily studied object within computer science \cite{barker2020scattered,DaySubsequenceUniversality2021,fleischmann2022nearly,halfon2017decidability,AbsentSubsequences,lothaire1997combinatorics,mateescu2004subword,simon2003words,tronicek2003common,zetzsche2016complexity} and beyond, with applications in a wide number of fields including bioinformatics \cite{han2020novel,shikder2019openmp}, database theory \cite{artikis2017complex}, and modelling concurrency \cite{shaw1978software}.
A recent survey of subsequence algorithms has been provided by Kosche et al. \cite{Kosche2022SubsequenceSurvey}, highlighting major results for problems on funding subsequences in words.

This paper considers \emph{$k$-subsequence universal} words.
A word $w$ is $k$-subsequence universal over an alphabet $\Sigma$ if $w$ contains every word of length $k$ over $\Sigma$ as a subsequence. 
These words were first defined by Karandikar and Schnoebelen \cite{karandikar2016height,schnoebelen2019height} as \emph{$k$-rich words}, however more recent work has used the term \emph{$k$-subsequence universality} \cite{barker2020scattered,DaySubsequenceUniversality2021,AbsentSubsequences}, which we will use here.
The study of these words follows from work on \emph{Simon’s congruence} \cite{simon1975piecewise}.
Informally, two words $w,v$ are $k$-congruent if $w$ and $v$ share the same set of subsequences of length $k$.
This relationship has been heavily studied \cite{fleischer2018testing,simon2003words,tronicek2003common,zetzsche2016complexity}, with a recent asymptotically optimal algorithm derived for testing if two words are $k$-congruent \cite{gawrychowski2021simons}.

Most relevant to this work are the papers by Barker et al. \cite{barker2020scattered}, and Day et al. \cite{DaySubsequenceUniversality2021}, directly addressing $k$-subsequence universal words.
In \cite{barker2020scattered}, the authors show that it is possible to determine, in linear time, if a word is $k$-subsequence universal or not, as well as the shortest $k$-subsequence universal prefix of a given word.
Additionally, they provide results showing that the minimal set of $\ell$-factors of a word $w$, $w_1 w_2 \dots, w_{\ell}$ such that $w_1 w_2 \dots w_{\ell}$ is $k$-subsequence universal, and the index $i$ such that $w^i$ is $k$-subsequence universal can be determined efficiently.

This is built on by \cite{DaySubsequenceUniversality2021}, in which the authors provide a set of algorithmic results for minimising the number of edit operations to transform a word into a $k$-subsequence universal word, providing results on \emph{insertions}, \emph{deletions}, and \emph{substitutions}.
They show that the minimum number of \emph{insertions} and \emph{substitutions} needed to transform a word $w$ into a $k$-subsequence universal word $w'$ can be done in $O(n k)$ time, assuming that $k < n$.
Additionally, they show that the number of \emph{deletions} needed to reduce the universality index (the maximum $k$ such that the word is $k$-subsequence universal) of a word to $k$ can be determined in $O(n k)$ time.

This paper is interested in providing algorithms for some of the basic operations on classes of words, \emph{counting}, \emph{ranking}, \emph{unranking}, and \emph{enumerating} for the class of $k$-subsequence universal words.
In providing these algorithms, we aim to expand the understanding of the space of $k$-subsequence universal words of a fixed length $n$.
We use $\mathcal{U}(n,k,\sigma)$ to denote the set of $k$-subsequence universal words of length $n$ over an alphabet of size $\sigma$ (assumed to be the alphabet $\{1,2,\dots,\sigma$\}).
The counting problem asks for the number of words in a given class.
The ranking problem takes as input a word $w$ and determines the number of words within the set which are lexicographically smaller than $w$.
The unranking problem is the inverse of the ranking problem, taking a rank $i$ and asking for the word in the set with the rank $i$.
Finally, the enumeration problem asks for the explicit outputting of every word within the set in some fixed order.
Each of these problems has been heavily studied for other classes of words, including cyclic words \cite{adamson2022ranking,Adamson2021,Fredricksen1978,gilbert1961symmetry,Kociumaka2014,Sawada2017} and Gray codes \cite{Fredricksen1978,Kociumaka2014,savage1997survey}.

\paragraph*{Our Results.}
This paper builds upon the existing body of work on $k$-subsequence universal words to build a stronger understanding of the space of $k$-subsequence universal words of a fixed length.
We provide a suite of algorithmic results for $k$-subsequence universal words of fixed length $n$.
We denote the set of all $k$-subsequence universal words over the alphabet $1,2,\dots,\sigma$ with length $n$ by $\mathcal{U}(n,k,\sigma)$.
In Section \ref{sec:counting}, we provide an an $O(n k \sigma)$ time algorithm for counting the size of $\mathcal{U}(n,k,\sigma)$.
In Section \ref{sec:ranking}, we use the observations from this counting algorithm to provide an $O(n k \sigma)$ time algorithm for ranking words in the set $\mathcal{U}(n,k,\sigma)$.
Finally, in Section \ref{sec:unranking} we provide an $O(n \sigma)$ time algorithm for unranking within the set $\mathcal{U}(n,k,\sigma)$, with $O(n k \sigma)$ time preprocessing.
We note this unranking algorithm directly provides an enumeration algorithm for the set $\mathcal{U}(n,k,\sigma)$ with $O(n \sigma)$ delay.

\paragraph*{Computational Model.}
In this paper, we assume the unit cost RAM computational model, in this case, equivalent to the unit cost word-RAM with word size $O(\log(N)\log(\sigma))$, where $N$ is the larger of the input or the output.
We note that this remains logarithmic relative to the number of $k$-subsequence universal words of length $n$, and thus the bits required to output the integer representation of the number of such words.
All our complexities can be readjusted into the unit cost RAM computational model with word size $O( \log(n) \log(\sigma))$ where $n$ is the size of the input, by applying a multiplicative factor of $O(n/\log (n))$ to the stated bounds.
We avoid a factor of $O(n^2/\log^2 (n))$ by noting that these algorithms only perform multiplications where at least one integer has size at most $\sigma$ or addition between integers of size at most $\sigma^n$.

\section{Preliminaries}
\label{sec:prelims}

We use the following notation.
Given a pair of natural numbers $m,n \in \mathbb{N}$, the notation $[m,n]$ denotes the ordered set $\{m, m+1, \dots, n \}$, or the empty set if $m > n$.
A word $w$ is an ordered sequence of symbols over some alphabet $\Sigma$.
The set of words of length $n$ over the alphabet $\Sigma$ is denoted $\Sigma^n$, and the set of all words over the alphabet $\Sigma$ by $\Sigma^*$.
The length of a word $w$ is denoted $|w|$.
The notation $w[i]$ is used to denote the $i^{th}$ symbol in the word $w$, and $w[i,j]$ is used to denote the contiguous sequence within $w$ corresponding to the word $w[i] w[i + 1] \dots w[j]$ (or the empty word $\emptyword$ if $i > j$).
A word $v$ is a \emph{factor} of a word $w$ if there exists some pair of indices $i,j$ such that $v = w[i,j]$.

We assume the alphabet $\Sigma = [1, \sigma]$ for some natural number $\sigma \geq 1$.
Given two words $w,v \in \Sigma^n$, $w$ is \emph{lexicographically smaller} than $v$ if there exists some index $i \in [1,n]$ such that $w[1,i - 1] = v[1,i - 1]$ and $w[i] < v[i]$.
Given two words $w,v \in \Sigma^*$, $v$ is a subsequence of $w$ if and only if there exists some series of indices $1 \leq i_1 < i_2 < \dots < i_{|v|} \leq |w|$ such that $v = w[i_1] w[i_2] \dots w[i_{|v|}]$.

\begin{definition}[$k$-subsequence universality]
    \label{def:k_universality}
    A word $w\in\Sigma^n$ is $k$-subsequence universal if and only if every word $v \in \Sigma^k$ is a subsequence of $w$.
    The set of words of length $n$ that are $k$-subsequence universal over an alphabet of size $\sigma$ is denoted $\mathcal{U}(n,k,\sigma)$.
\end{definition}


The \emph{subsequence universality index} of a word $w$ is the largest value $k$ such that $w$ is $k$-subsequence universal.
In order to determine if a word is $k$-subsequence universal, we use \emph{arch-factorisations}, first introduced by Hebard \cite{hebrard1991algorithm}.
Informally, an \emph{arch} of a word is a minimal length factor containing each symbol in the alphabet $\Sigma$ at least once.
For the remainder of this paper, we use the following formal definition:

\begin{definition}[Arches]
    \label{def:arch}
    An \emph{Arch} over the alphabet $\Sigma = \{1,2,\dots, \sigma\}$ is a word containing every symbol in $\Sigma$ at least once.
    The \emph{universal subsequence} of an arch $v$ is the set of indices $(i_1, i_2, \dots, i_{\sigma})$ satisfying the following:
    \begin{itemize}
        \item $v[i_1], v[i_2], \dots, v[i_{\sigma}]$ contains every symbol in $\Sigma$ exactly once.
        \item The index $i_j$ is the first position in $v$ where the symbol $v[i_j]$ appears.
        \item The index $i_{\sigma}$ is the last position in $v$.
    \end{itemize}
    Any symbol not in the universal subsequence is called a \emph{free symbol}.
\end{definition}

We note that this definition of an arch corresponds to a $1$-subsequence universal word where the last symbol is unique, i.e. it does not appear anywhere else in the word.

\begin{definition}[\cite{hebrard1991algorithm}, Arch Factorisations]
    \label{def:arches}
    The arch-factorisation of a word $w \in \Sigma^n$ with a universality index of $k$, is a set of factors $\{w_1, w_2, \dots, w_k, v\}$, denoted $Arch(w)$, such that $w_i$ is a factor of $w$ that is an arch for every $i \in [1,k]$, $v$ is a suffix of $w$ that does not contain any arch as a factor, and $w_1 w_2 \dots w_k v = w$.
\end{definition}

\begin{figure}
    \centering
    \begin{align*}
        &w = {\color{red}11234},{\color{blue}4321},{\color{green}22314},{\color{orange} 33214}, 4\\
        &v = {\color{red} 12234}, {\color{blue} 323134}, {\color{green} 11234}, 4412
    \end{align*}
    \caption{An example of the arch-factorisation of two words $w,v \in \Sigma^{20}$ where $\Sigma = \{1,2,3,4\}$.
    Each arch (or the ending suffix) is separated by a comma, and highlighted in a seperate colour.
    Note that $w$ is $4$-subsequence universal while $v$ is only $3$-subsequence universal, despite sharing the same Parikh vector $(5,5,5,5)$.}
    \label{fig:factorisation}
\end{figure}

An example of this factorisation is given in Figure \ref{fig:factorisation}.
Day et al. \cite{DaySubsequenceUniversality2021} expanded upon Definition \ref{def:arches} to show that a word is $k$-subsequence universal if and only if there exists an arch-factorisation of $w$ containing at least $k$-arches, and further, that such a factorisation can be computed in time linear to the length of the word.

\begin{theorem}[\cite{DaySubsequenceUniversality2021}]
    \label{thm:arch}
    A word $w \in \Sigma^n$ is $k$-subsequence universal over $\Sigma$ if and only if $Arch(w)$ contains at least $k$ arches. Further, $Arch(w)$ can be computed in $O(n)$ time.
\end{theorem}

In this paper, we use the following technical Lemma from \cite{DaySubsequenceUniversality2021}.

\begin{lemma}[\cite{DaySubsequenceUniversality2021}]
	\label{lem:alph_counting}
	Let $\Delta(w, i, j)$ denote the number unique symbols in $w[i,j] $ for some $w \in \Sigma^n$.
	We can compute in $O(n)$ the values of $\Delta(w,1, j)$ for every $j \in [1,n]$. 
\end{lemma}

We combine Theorem \ref{thm:arch} and Lemma \ref{lem:alph_counting} to make the following observation.

\begin{observation}
	\label{obs:alph_and_arches}
	Let $w \in \Sigma^n$ be a $k$-subsequence universal word with the arch-decomposition $w_1, w_2, \dots, w_k, v$, and further let $\mathcal{A} = \{ A_1, A_2, \dots, A_k\}$ denote the set of indices where $A_{\ell} = 1 + \sum_{i \in [1,\ell - 1]} |w_i|$, i.e. the set of indices in $w$ corresponding to the first position of an arch in $Arch(w)$.
	Then, the values of $\Delta(w, A_{\ell}, i_{\ell})$ can be computed in $O(n)$ time for every $\ell \in [k]$ and  $i_{\ell} \in [A_{\ell},A_{\ell + 1} - 1]$, where $\Delta(w, A_{\ell}, i_{\ell})$ denotes the number unique symbols in $w[A_{\ell},i_{\ell}]$.
\end{observation}

Using this notation, we provide a formal definition of the ranking and unranking problems as considered in this paper.
The \emph{rank} of a word $w$ within an ordered set of words $\mathcal{S}$ is the number of words in $\mathcal{S}$ that are smaller than $w$ under the ordering of the set.
In this paper, we assume that the set of $k$-subsequence universal words is ordered lexicographically, and therefore the rank of a word $w$ in the set $\mathcal{U}(n,k,\sigma)$ is the number of words in $\mathcal{U}(n,k,\sigma)$ that are lexicographically smaller than $w$.
The \emph{ranking} problem takes as input a word $w$ integer triple $n,k,\sigma \in \mathbb{N}$ such that $n \geq k \sigma$, and returns the number of words in $\mathcal{U}(n,k,\sigma)$ lexicographically smaller than $w$.
The \emph{unranking} problem is conceptually the inverse of the ranking problem.
Given an integer $i \in [1, |\mathcal{S}|]$, the unranking problem asks for the word in $\mathcal{S}$ with a rank of $i$.
In this paper, the unranking problem takes as input a rank $i$ and integer triple $n,k,\sigma \in \mathbb{N}$ such that $n \geq k \sigma$, and returns the word in $\mathcal{U}(n,k,\sigma)$ with a rank of $i$.

\section{Counting Arches and $k$-subsequence universal words}
\label{sec:counting}

First, we present a tool for counting the size of $\mathcal{U}(n,k,\sigma)$, i.e. number of $k$-subsequence universal words of length $n$ over the alphabet $\Sigma = \{1,2, \dots, \sigma\}$. 
As well as being an interesting result in and of itself, this provides the foundation for our tools for both ranking and unranking.

This section is split into two sections.
First, we provide formulae for counting the number of arches, $0$-subsequence universal words, and 1-universal words of length $n$ over an alphabet of size $\sigma$.
Second, we provide a recursive technique to count the number of $k$-subsequence universal words of length $n$ over an alphabet of size $\sigma$.

%
%

\subsection{Arches, $0$-subsequence universal and 1-subsequence universal words}

We first consider how to count the number of arches, $0$-subsequence universal and 1-subsequence universal word of length $n$.
We note that these three special cases are closely interlinked.
First, note that any word that is not $0$-subsequence universal must be at least $1$-subsequence universal.
Therefore, the number of $1$-subsequence universal words is equal to the number of words minus the number of $0$-subsequence universal words.
Similarly, the number of $n$-length arches is equal to the number of $(n - 1)$-length $1$-subsequence universal words over an alphabet of size $\sigma - 1$, multiplied by $\sigma$.
We start with $0$-subsequence universal words.

\begin{lemma}
\label{lem:0_universal}
The number of $n$-length $0$-subsequence universal words over an alphabet $\Sigma = \{1,2, \dots, \sigma\}$ is given by:

$$\sum\limits_{i \in [1,\sigma]} (-1)^{i + 1} \genfrac(){0pt}{0}{\sigma}{i} (\sigma - i)^n.$$
\end{lemma}

\begin{proof}
    Let $\varsigma \subseteq \Sigma$ be a $i$-length alphabet.
    Note first that the number of $n$-length words over $\varsigma$ is given by $i^{n}$, and further, as there are $\genfrac(){0pt}{2}{\sigma}{i}$ such alphabets, the total number of words over \emph{any} $i$-length alphabet is given by $\genfrac(){0pt}{0}{\sigma}{i} i^n$.
    Observe that any string in $\varsigma^n$ is also in $(\varsigma \cup \{x\})^n$, for some $x \in \Sigma \setminus \varsigma$, and further, there are $\sigma - i$ such $i + 1$-length alphabets containing every symbol in $\varsigma$.
    More generally, there are $\genfrac(){0pt}{2}{\sigma - i}{j}$ alphabets of size $i + j$ containing every symbol in the $i$-length alphabet $\varsigma$.
    Therefore, taking the sum of $n$-length words in all $i$-length alphabets, given by $\genfrac(){0pt}{0}{\sigma}{i} (i)^n$, will also count every word in a $j < i$-length language $\genfrac(){0pt}{0}{\sigma - j}{i - j}$ times.
    Combining this with the well-known binomial coefficient identities gives the equation for the total number of unique words in any alphabet in the set $\{\Sigma \setminus \{x\} \mid x \in \Sigma\}$ as:

    $$\genfrac(){0pt}{0}{\sigma}{1} (\sigma - 1)^{n - 1} - \genfrac(){0pt}{0}{\sigma}{2} (\sigma - 2)^{n - 1}+ \genfrac(){0pt}{0}{\sigma}{3} (\sigma - 3)^{n - 1} \dots (-1)^{\sigma + 1}$$
    $$= \sum\limits_{i \in [1,\sigma]} (-1)^{i + 1} \genfrac(){0pt}{0}{\sigma}{i} (\sigma - i)^n$$
\end{proof}

Using Lemma \ref{lem:0_universal}, the counting of $n$-length arches and $1$-subsequence universal words follows directly.

\begin{corollary}
    \label{col:1_universal}
    The number of $n$-length $1$-universal words over an alphabet $\Sigma = \{1,2, \dots, \sigma\}$ is given by:
    $$\sigma^n - \sum\limits_{i \in [1,\sigma]} (-1)^{i + 1} \genfrac(){0pt}{0}{\sigma}{i} (\sigma - i)^n$$
\end{corollary}

\begin{corollary}
    \label{col:arches}
    The number of $n$-length Arches over an alphabet $\Sigma = \{1,2, \dots, \sigma\}$ is given by:
    $$\sigma(\sigma - 1)^{n - 1} - \sum\limits_{i \in [2,\sigma]} (-1)^{i} i\genfrac(){0pt}{0}{\sigma}{i} (\sigma - i)^{n - 1}.$$
\end{corollary}

\subsection{Counting $k$-subsequence universal words}

To count $k$-subsequence universal words with an arbitrary value of $k$, we employ a recursive approach.
The high-level idea is to count the number of suffixes of $k$-subsequence universal words sharing a given prefix $v$.
Let $\mathcal{S}(v)$ be the set of words of length $n - \mid v \mid$ such that for every word $u \in \mathcal{S}(v)$, the word $v u$ is a $k$-subsequence universal word.
Let $Arch(v) = v_1, v_2, \dots, v_{\ell}, v'$ be the arch factorisation of $v$, or the set of the first $k$ arches of $v'$.
In order to count the size of $\mathcal{S}(v)$, we observe that every word $u \in \mathcal{S}(v)$ must contain a prefix $u'$ such that $v' u'$ is an arch and the suffix $u[|u'| + 1, |u|]$ must contain $k - \ell - 1$ arches.
Our recursive approach is based on the observation that the size of $\mathcal{S}(v)$ is equal to the size of $\bigcup_{x \in \Sigma} \mathcal{S}(v x)$.
This leaves two major problems: determining whether or not the set $\mathcal{S}(v x)$ is empty, and ensuring that the total size of $\mathcal{S}(v)$ can be computed without having to explicitly check $\mathcal{S}(v w)$ for every suffix $w \in \Sigma^{n - |v|}$.

We solve these problems by introducing a new function, $\CountSuffixes(q, m, c)$ (\textbf{C}ount \textbf{S}uffixes) such that $|\mathcal{S}(v)| = \CountSuffixes(q, m, c)$ where:
\begin{itemize}
    \item $q$ is the number of unique symbols in $v'$.
    \item $m$ is the number of free symbols in every word $u\in \mathcal{S}(v)$, i.e. the number of symbols in $u$ that do not belong to any universal subsequence of the first $k - \ell - 1$ arches of $u[|u'| + 1, |u|]$ or in the universal subsequence of $u'$ in the word $v' u'$.
    \item $c$ is the minimum number of arches in $v' u$, equal to $k - \ell$.
\end{itemize}
The value of $\CountSuffixes(q,m,c)$ is determined in a recursive manner.
We first provide the base cases.
If $m = 0$, then every remaining symbol must be in the universal subsequence for one of the remaining $c$ arches, giving $\CountSuffixes(q, 0, c) = (\sigma - q)!(\sigma!)^{c - 1}$.
On the other hand, if $c = 0$, then the remaining symbols can be chosen arbitrarily from $\Sigma$, giving $\CountSuffixes(q,m,0) = \sigma^m$.
Assuming both $c$ and $m$ are greater than $0$, then the value of $\CountSuffixes(q,m,k)$ is determined recursively.
If $q = \sigma$, then the next symbol must be the first symbol of the $(k - c)^{th}$ arch of the word.
As there are $\sigma$ such possible symbols, followed by one of  $\CountSuffixes(1, m, c - 1)$ suffixes, the value of $\CountSuffixes(\sigma, m, c)$ is $\sigma  \CountSuffixes(1, m, c - 1)$.
Otherwise, the next symbol can either be one of the $q$ symbols already in the universal subsequence of the current arch, or one of the $\sigma - q$ symbols not in the universal subsequence, giving $\CountSuffixes(q, m, c) =  (\sigma - q)\CountSuffixes(q + 1, m, c) + q \CountSuffixes(q, m - 1, c)$.
Putting this together, the function $\CountSuffixes(q,m,c)$ can be defined as:


\[
	\CountSuffixes(q, m, c) = \begin{cases}
		(\sigma - q)! (\sigma!)^{c - 1} & m = 0\\
		\sigma^{m} & c = 0\\
		\sigma \CountSuffixes(1, m, c - 1) & q = \sigma, m > 0, c > 0\\
		(\sigma - q) \CountSuffixes(q + 1, m, c) + q \CountSuffixes(q, m - 1, c) & q > 0, m > 0, c > 0
	\end{cases}
\]

\begin{lemma}
    \label{lem:suffixToCountSuffix}
    Let $\mathcal{S}(v)$ denote the set of suffixes such that for every $u \in \mathcal{S}(v)$, the word $v u$ is an $n$-length $k$-subsequence word.
    Further, let $\ell$ be the number of complete arches in $v$, and let $v'$ be the suffix of $v$ such that $Arch(v) = v_1 v_2 \dots v_{\ell - 1} v_{\ell} v' = v$ where $v_i$ is the $i^{th}$ arch in the arch factorisation for $v$.
    Then, the size of $\mathcal{S}(v)$ is equal to $\CountSuffixes(q, m, k - \ell)$ where $q$ is the number of unique symbols in $v'$, and $m = n - (|v| + \sigma (k - \ell - 1))$.
\end{lemma}

\begin{proof}
    We will assume, for notational simplicity, that if $v' = \emptyword$, then $q = \sigma$.
    We start with the base cases.
    If $m = 0$, every symbol in the suffixes of $\mathcal{S}(v)$ must be in the universal subsequence of one of the remaining arches.
    As there are $q$ symbols in $v'$, there are $(\sigma - q)!$ possible ways of extending $v'$ to become an arch, and $\sigma!$ arches of length $\sigma$, the total number of suffixes in $\mathcal{S}(v)$ is  $(\sigma - q)! (\sigma!)^{k - \ell - 1} = (\sigma - q)(\sigma!)^{c - 1}$.
    Alternatively, if $\ell \geq k$ (meaning that $v$ is already a $k$-subsequence universal word), then every remaining symbol in the suffix is a free symbol, and as such there are no constraints on the contents of the suffix.
    Therefore in this case, there are $\sigma^m$ suffixes in $\mathcal{S}(v)$.
    Note that in the case $m = 0$ and $\ell \geq k$, both of these formulae return 0, corresponding to the empty word.

    In the general case, assume that the size of $\mathcal{S}(v x)$ is equal to $\CountSuffixes(q', m', c')$, where $q'$ is the number of unique symbols in $v' x$, $m'$ is equal to $n - (|v| + 1) - (q + \sigma(k - \ell))$, and $c'$ is $k - \ell$.
    Observe that the total number of suffixes in $\mathcal{S}(v)$ is equal to $\sum_{x \in \Sigma} |\mathcal{S}(v x)|$.
    If $q = \sigma$, then the next symbol must belong to the universal subsequence of the next arch, equal to the value of $\CountSuffixes(1, m, k - \ell - 1)$.
    If $x$ is one of the $q$ symbols that have already appeared in $v'$, then the size of $\mathcal{S}(v x)$ is $\CountSuffixes(q, m - 1, c)$.
    Otherwise, the size of $\mathcal{S}(v x)$ is $\CountSuffixes(q + 1, m, c)$.
    As there are $q$ unique symbols in $v'$, the number of suffixes is given by the sum $q \CountSuffixes(q, m - 1, c) + (\sigma - q) \CountSuffixes(q + 1, m, c)$.
    Hence the size of $\mathcal{S}(v)$ is given by $\CountSuffixes(q, m, k - \ell)$.
\end{proof}

\begin{lemma}
	\label{lem:countsuffixes}
 	The values of $\CountSuffixes(q,m,c)$ can be computed for every $q \in [1,\sigma], m \in [0,n], c \in [0,k]$ in $O(n k \sigma)$ time.
\end{lemma}

\begin{proof}
	The correctness follows for the arguments above.
	We assume that the values of  $(\sigma - q)! (\sigma!)^{c}$ have been precomputed for every $q \in [1,\sigma], c \in [0,k]$, requiring $O(k \sigma)$ time, and the values of $\sigma^m$ have been precomputed for every $m \in [0,n]$ requiring $O(n)$ time.
	To determine the time complexity, note that the value of $\CountSuffixes(q,m,c)$ can be computed in constant time assuming that the values of $\CountSuffixes(q + 1,m,c)$, $\CountSuffixes(1,m,c - 1)$, and $\CountSuffixes(q,m- 1,c)$ have already been computed.
	As the base cases of $c = 0$, and $m = 0$ can be computed in constant time, and for every other case the values of $m$ and $(\sigma - q) + \sigma c$ are monotonically decreasing, the values of $\CountSuffixes(q,m,c)$ can be computed for every $q \in [1,\sigma], m \in [0,n], c \in [0,k]$ in a dynamic manner, starting with the base cases, and proceeding in increasing value of $m$, $\sigma - q$ and $c$.
	We note that the order in which $m,q$ and $c$ are incremented is irrelevant provided $\CountSuffixes(q + 1,m,c)$, $\CountSuffixes(1,m,c - 1)$, and $\CountSuffixes(q,m- 1,c)$ are computed before $\CountSuffixes(q ,m,c)$.
	Therefore, the total time complexity of computing the values of $\CountSuffixes(q,m,c)$ for every $q \in [1,\sigma], m \in [0,n], c \in [0,k]$, is $O(n k \sigma)$.
\end{proof}

Note that the number of $k$-subsequence universal words is equal to the number of words in the set $\mathcal{S}(\emptyword)$, i.e. the number of $n$-length words with an empty prefix.
Therefore, from Lemma \ref{lem:countsuffixes}, it follows that the size of $\mathcal{U}(n,k,\sigma)$ can be computed by counting the size of $\mathcal{S}(\emptyword)$, equivalent to evaluating $\sigma \CountSuffixes(1,n - (k \sigma), k)$.
Theorem \ref{thm:counting} follows from this observation.

\begin{theorem}
    \label{thm:counting}
    The size of $\mathcal{U}(n, k, \sigma)$ can be computed in $O(n k \sigma)$ time.
\end{theorem}

\section{Ranking}
\label{sec:ranking}

Using the counting techniques outlined in Section \ref{sec:counting}, we can now rank a given word $w \in \Sigma^n$ amongst the set of $n$-length $k$-subsequence universal words.
This is done in an iterative manner.
For each $i \in [1,n]$, we count the number of words of $n$-length $k$-subsequence universal words with the prefix $w[1,i - 1] x$, where $x$ is some symbol lexicographically smaller than $w[i]$.
Taking the sum of such words for every $i \in [1,n]$ gives the total number of $n$-length $k$-subsequence universal words that are lexicographically smaller than $w$.
By taking the sum of such words for each prefix, the total number of $n$-length $k$-subsequence universal words that are lexicographically smaller than $w$ can be computed.

Let $Arch(w) = w_1, w_2, \dots, w_m v$.
The first key observation is that given any word $u = w_1 w_2 \dots w_i u'$, for some $i \leq k$, $u$ is $k$-subsequence universal if and only if $u'$ is $(k - i)$-subsequence universal.
Secondly, given a word $s = w_1 w_2 \dots w_i[1:j] s'$, $s$ is $k$-subsequence universal if and only if $w_i[1:j] s'$ is $(k - i + 1)$-subsequence universal.

\paragraph*{Preprocessing.}
In order to make our ranking algorithm more efficient, we first provide an overview of the preprocessing that is performed before the main ranking algorithm.
Let $Arch(w) = w_1, w_2, \dots, w_k, v$.
Using the notation from Observation \ref{obs:alph_and_arches}, let $A_1, A_2, \dots, A_k$ be the indices such that $A_{\ell}$ corresponds to the first position in $w$ at which the arch $w_{\ell}$ appears in $w$, i.e. $w_{\ell} = w[A_{\ell}, A_{\ell + 1} - 1]$.
Further, let $\Delta(w,A_{\ell},i)$ be the number of unique symbols in the $i$-length prefix of $w_{\ell}$. 
We assume that the values of $\Delta(w,A_{\ell},i_{\ell})$ have been computed for every $i \in [A_{\ell},A_{\ell + 1} -1]$, $\ell \in [0,k]$.

In order to count the number of free symbols within each suffix of $w$, let $m$ be an $n$-length array such that $m[i]$ contains the number of free symbols in the $i$-length suffix of $w$.
The values of $m[i]$ are computed by starting with $i = n$, and working in decreasing value of $i$.
Note that the value of $m[n]$ is equal to $0$.
In the general case, the value of $m[i]$ is either $m[i + 1]$, if $w[i]$ belongs to the universal subsequence of some arch in the arch decomposition, of $m[i + 1] + 1$ otherwise.
Letting $\ell$ be the index such that $A_{\ell} \leq i < A_{\ell + 1}$, note that if $w[i]$ is in the universal subsequence of $w_{\ell}$, then $\Delta(w,A_{\ell}, i) = \Delta(w,A_{\ell}, i - 1) + 1$, otherwise $\Delta(w,A_{\ell}, i) = \Delta(w,A_{\ell}, i - 1)$.
Hence using the previous computation, the values of $m[i]$ can be determined in $O(n)$ time for every $i \in [1,n]$.

In order to determine the number of symbols smaller than $w[i]$, an additional $n$-length array $l$ such that $l[i]$ contains the set of symbols that appear between $A_{\ell}$ and $i$ in $w$, where $\ell$ is the index such that $A_{\ell} \leq i < A_{\ell + 1}$.
This complements $m$ by ensuring providing a quick method of checking if a given symbol $x$ has already been used by $w_{\ell}[1, i + 1 - A_{\ell}]$.
The array $l$ is computed in $O(n \sigma)$ time as follows.
For each $i \in [1,n]$, note that the value of $l[i]$ is either $l[i - 1] \cup \{w[i]\}$, if $i \neq A_{\ell}$ for every $\ell \in [0,k]$, or $\{w[i]\}$ otherwise.
By storing each array as a $\sigma$-length binary vector, requiring at most $O(n \sigma)$ time to initialise, the values of $l[i]$ can be computed for every $i \in [1,n]$ in $O(n)$ time.
Finally, we assume that the value of $\CountSuffixes(q,m,c)$ has been precomputed for every $q \in [1, \sigma], m \in [n]$ and $c \in [k]$.

\paragraph*{Ranking.}
We now have the tools we need to rank the input word $w \in \Sigma^n$.
We note that $w$ does not have to be a $k$-subsequence universal word, allowing this tool to be used in a more general setting.
At a high level, our approach is to take each prefix of $w$, $w[1,i]$, and count the number of words in $\mathcal{U}(k,n,\sigma)$ that are lexicographically smaller than $w$ with the prefix $w[1,i] x$, where $x < w[i + 1]$.
By taking the sum of such words for each prefix of $w$, the total number of words smaller than $w$ can be determined.

Let $\mathcal{R}(i) = \{u \in \mathcal{U}(k,n,\sigma) \mid u < w, u[1,i] = w[1,i], u[i + 1] < w[i + 1]\}$, and let $\ell$ be the index such that $A_{\ell} \leq i < A_{\ell + 1}$.
Note that the number of possible values for the symbol $u[i + 1]$ is equal to $w[i + 1] -1$.
Further, the number of words in $\mathcal{R}(i)$ with the prefix $w[1,i] x$ for some fixed $x < w[i + 1]$ is equal to either $\CountSuffixes(\Delta(w, A_{\ell},i) + 1, m[i], k + 1 - \ell )$, if $x$ is in the universal subsequence of $w_{\ell}$ or $\CountSuffixes(\Delta(w, A_{\ell},i), m[i] - 1, k + 1 - \ell )$ otherwise.
Recall that the array $l$ contains at position $i$ the set of unique symbols in the factor of $w$ $w[A_{\ell}, i]$.
Therefore, the size of $\mathcal{R}(i)$ can be computed with this following sum:

\[|\mathcal{R}(i)| =
    \sum_{x \in [1,w[i + 1] - 1]}\begin{cases}
        \CountSuffixes(\Delta(w, A_{\ell},i) + 1, m[i], k + 1 - \ell ) & x \in l[i]\\
        \CountSuffixes(\Delta(w, A_{\ell},i), m[i] - 1, k + 1 - \ell ) & x \notin l[i]
    \end{cases}
\]

Using $\mathcal{R}(i)$, rank of $w$ in the set $\mathcal{U}(n,k,\sigma)$, denoted $rank(w)$, is given by:

\[
rank(w) = \sum_{i \in [0, n - 1]} |\mathcal{R}(i)|
\]

\begin{theorem}
    The rank of a given word $w \in \Sigma^n$ can be determined in $O(n k \sigma)$ time.
\end{theorem}

\begin{proof}
    Observe that for any word of the form $w[i] x v$ to be $k$-subsequence universal, the suffix $v$ must belong to $\mathcal{S}(w[i] x)$.
    Let $\ell$ be the index such that $A_{\ell} \leq i < A_{\ell + 1}$, $q$ be the number of unique symbols in $w[A_{\ell}, i]$ and $m$ the number of free symbols following $w[1,i] x$.
    Note that the number of possible values of $v$ is either $\CountSuffixes(q + 1, m, k - \ell)$, if $x$ is not in the universal subsequence of $w[A_{\ell}, i]$, or $\CountSuffixes(q, m - 1, k - \ell)$ if $x$ has already appeared in $w[A_{\ell}, i]$.
    Using the list $l$, it can be determined in constant time if the symbol $x$ appears in $w[A_{\ell}, i]$.
    By extension, the total number of $n$-length $k$-subsequence universal words with the prefix $w[i] x$ can be computed in constant time, assuming that the values of $\CountSuffixes(q,m,k - \ell)$ has been precomputed, and hence the value of $\mathcal{R}(i)$ can be computed in $O(\sigma)$ time.
    As there are $n$ possible prefixes of $w$, the total rank of $w$ within $\mathcal{U}(n,k,\sigma)$ can be computed in $O(n \sigma)$ time after $O(n \sigma k)$ preprocessing.
\end{proof}

\section{Unranking}
\label{sec:unranking}

We complement our counting and ranking techniques by showing how to unrank $n$-length $k$-subsequence universal words.
Note that an efficient unranking technique may be used as an effective tool to enumerate the set of all $k$-subsequence universal words.
We assume that the values of $\CountSuffixes(q,m,k)$ have been precomputed for every $q \in [1,\sigma], m \in [0,n]$, and $c \in [0,k]$.

Our unranking processes operates in an iterative manner.
Let $w$ be the word of rank $i$ that is being unranked.
Starting with $j = 1$, the value of $w[j]$ is computed by counting the number of $n$-length $k$-subsequence universal words with the prefix $w[1,j - 1] x$, for $x \in \Sigma$ starting with $x = 1$.
The value of $x$ is increased until the number of words with a prefix smaller than or equal to $w[1,j - 1] x$ is greater than $i$.
Once this value of $x$ has been computed, $w[j]$ is set to $x - 1$, and the algorithm proceeds to compute the value of $w[j + 1]$.

\begin{theorem}
    \label{thm:unranking}
    The $k$-subsequence universal word $w$ of length $n$ with a rank of $i$ can be determined in $O(n \sigma + n k \sigma)$ time.
\end{theorem}

\begin{proof}
    Starting with $w[1]$, note that the number of words with the prefix $x$, for any $x \in \Sigma$, is given by $\CountSuffixes(1, n - (k \sigma), k)$.
    Further, any word with the first symbol $x$ has a rank in the range $(x - 1) \CountSuffixes(1, n - (k \sigma), k) + 1$ to $x \CountSuffixes(1, n - (k \sigma), k)$.
    Therefore the value of $w[1]$ is the value of $x$ such that $(x - 1) \CountSuffixes(1, n - (k \sigma), k) < i \leq x\CountSuffixes(1, n - (k \sigma), k)$.

    More generally, let $t(j)$ be the smallest rank of words with the prefix $w[1,j]$, determined by the sum:
    \[
        t(j) = \sum\limits_{\ell \in [1,j]} \sum\limits_{x \in [1,w[\ell] - 1 ]} |\mathcal{S}(w[1,\ell - 1] x)  | = t(j - 1) + \sum_{x \in [1,w[j] - 1 ]} |\mathcal{S}(w[1,j - 1] x)|.
    \]
    Note that the value of $t(j)$ can therefore be computed in $O(\sigma)$ time using $t(j - 1)$ and the values of $\CountSuffixes(q,m,c)$.
    The value of $w[j + 1]$ is, therefore, the symbol $x$ such that $t(j) + \sum_{y \in [1,x - 1]} |\mathcal{S}(w[1,j] y)|< i \leq t(j) + \sum_{y \in [1,x]} |\mathcal{S}(w[1,j] y)|$, and further can be computed in $O(\sigma)$ time, giving the total time complexity of the unranking of $w$ as $O(n \sigma)$ after $O(n \sigma k)$ preprocessing.
\end{proof}

\begin{corollary}
    The set of $k$-subsequence universal words of length $n$ can be output explicitly with $O(n \sigma)$ delay after $O(nk\sigma)$ preprocessing.
\end{corollary}

\begin{proof}
    Following Theorem \ref{thm:unranking}, each index $i \in [1,|\mathcal{U}(n, k, \sigma)|]$ can be unranked in $O(n \sigma)$ time after at most $O(n \sigma k)$ preprocessing.
    Hence the set $\mathcal{U}(n,\sigma,k)$ can be enumerated with $O(n \cdot \sigma)$ delay after $O(nk\sigma)$ preprocessing.
\end{proof}

\section{Conclusion}

In this paper, we provided new tools for understanding the space of $k$-subsequence universal words.
Notably, we have shown how to count, rank, unrank, and enumerate these words with efficient algorithms for words of fixed length.
We note that all of these algorithms can be extended to the setting of words of length at most $n$.
We see two key open questions asked in this paper.
First, if there is a general formula for counting the number of $n$-length $k$-subsequence universal words.
Indeed, such a formula may allow for a speed up for the preprocessing of the ranking, unranking, and enumeration algorithms, if it can be extended to count the size of $\mathcal{S}(v)$ efficiently.
Secondly, if there is an enumeration algorithm outputting every word in $\mathcal{U}(n,k,\sigma)$ with at most O$(n)$ delay after polynomial-time preprocessing.

The author thanks the Leverhulme Trust for funding this research via the Leverhulme Research Centre for Functional Materials Design.
Further, the author would like to thank the reviewers for their helpful comments that have improved the readability of this paper.

\printbibliography

\end{document}